\documentclass[11pt,twoside]{amsart}

\usepackage{amsmath}
\usepackage{amsthm}
\usepackage{amsfonts, amssymb}
\usepackage{mathrsfs}
\usepackage[all,line]{xy}
\usepackage{url}
\usepackage{float} 

\usepackage{epstopdf}
\usepackage{bm}
\usepackage{algorithm}
\usepackage{algpseudocode}
\usepackage{latexsym}
\usepackage{graphicx}
\usepackage{xurl}      
\usepackage{hyperref}

\theoremstyle{plain}
\newtheorem{lemma}{Lemma}[section]
\newtheorem{proposition}[lemma]{Proposition}
\newtheorem{thm}[lemma]{Theorem}

\theoremstyle{remark}

\newtheorem{remark}[lemma]{Remark}

\theoremstyle{definition}
\newtheorem{definition}[lemma]{Definition}

\def\R{\mathbb{R}}

\newcommand{\scp}{\mbox{$\ \searrow \! \! \! \! \searrow \! \! \! \! \ \ \ $}}

\newcommand{\esc}{\mbox{$\ \searrow \! \! \! \! \searrow \! \! \! \! ^e \ $}}

\newcommand{\nd}{N_{\delta}}

\begin{document}

\title[$\delta$-cores and TDA]{\bm{$\delta$}-core subsampling, strong collapses and TDA}
\author[E.G. Minian]{El\'\i as Gabriel Minian}

\address{Departamento  de Matem\'atica-IMAS\\
 FCEyN, Universidad de Buenos Aires\\ Buenos
Aires, Argentina.}

\email{gminian@dm.uba.ar}

\thanks{Researcher of CONICET. Partially supported by grant UBACYT 20020190100099BA}

\begin{abstract}
We introduce a subsampling method for topological data analysis based on strong collapses of simplicial complexes. Given a point cloud and a scale parameter $\delta$, we construct a subsampling that preserves both global and local topological features while significantly reducing computational complexity of persistent homology calculations. We illustrate the effectiveness of our approach through experiments on synthetic and real datasets, showing improved persistence approximations compared to other subsampling techniques.

\end{abstract}



\maketitle

\section{Introduction}

Topological Data Analysis (TDA) applies methods and tools from algebraic topology to analyze datasets. One of the most widely used techniques in this area is persistent homology, which aims to infer the homology groups of an underlying (unknown) space from a sample $X$ by examining the data at all possible scales \cite{C09,EH08,ELZ02,G14,OP17,RB19,Z09,ZC05}.
Currently, TDA and persistent homology have found diverse applications in many scientific fields including medicine, biology, engineering, physics, chemistry, and industry \cite{C09, RB19, RF24}. One important consideration is the computational complexity of persistent homology: the standard matrix reduction algorithm has complexity $O(s^3)$, where $s$ is the number of simplices in the Vietoris-Rips filtration \cite{EH08, B21, M05}. When analyzing a dataset of several thousand points, the number of simplices in the filtration used for computing persistent homology can reach several million if we wish, for example, to compute homology up to degree 2. It is therefore necessary to have effective subsampling methods that preserve the topology of the data. Some known methods exist in this direction, such as farthest point sampling (FPS) \cite{dT04}, landmark selection with witness complexes \cite{CdS04}, graph sparsification methods \cite{S13} and, more recently, adaptive approximation techniques \cite{HJV24} and outlier-robust subsampling techniques \cite{S23}.

In this article we propose a new subsampling method based on the concept of \textit{strong collapse}, a reduction technique in the context of simplicial complexes that we introduced some years ago in collaboration with J. Barmak \cite{BM11}. Strong collapses allow the study of the topology of polyhedra through minimal subpolyhedra (the \textit{cores}) that preserve the homotopy type of the original complex. A strong collapse removes a dominated vertex from the simplicial complex that is, in some sense, redundant, along with all simplices containing it (see Section \ref{sectionstrong} below). The core of a simplicial complex is obtained by iteratively removing dominated vertices until reaching a space with no dominated vertices. This space preserves the (strong) homotopy type of the original space but generally contains far fewer simplices. In the particular case of Vietoris-Rips complexes (used for persistent homology computations), the reduction in simplices can in some cases exceed 80\%. Table~\ref{tab:reduction} shows the simplex reduction when computing the cores of the $3$-skeleta of the Vietoris-Rips complexes at different scales $\varepsilon$ for a sample of 500 points randomly chosen on the sphere $S^2$ in $\mathbb{R}^3$ (with Euclidean distance).

\begin{table}[h]
	\centering
	
	\begin{tabular}{rrrrl}
		\hline
		\# & $\varepsilon$\ \ \ \ \  & $\ \ |\mathrm{VR}(X,\varepsilon)|$ & $\ \ \ |\mathrm{core}|$ & Reduction \\
		\hline
		1  & 0.000000 & 700      & 700      & 0.0\%  \\
		2  & 0.042857 & 721      & 681      & 5.5\%  \\
		3  & 0.085713 & 869      & 563      & 35.2\% \\
		4  & 0.128570 & 1303     & 351      & 73.1\% \\
		5  & 0.171427 & 2366     & 248      & 89.5\% \\
		6  & 0.214283 & 5201     & 1247     & 76.0\% \\
		7  & 0.257140 & 11598    & 3745     & 67.7\% \\
		8  & 0.299997 & 27280    & 9945     & 63.5\% \\
		9  & 0.342853 & 62047    & 23775    & 61.7\% \\
		10 & 0.385710 & 126608   & 52750    & 58.3\% \\
		11 & 0.428567 & 241405   & 94834    & 60.7\% \\
		12 & 0.471423 & 431876   & 205132   & 52.5\% \\
		13 & 0.514280 & 750750   & 291772   & 61.1\% \\
		14 & 0.557137 & 1228758  & 415724   & 66.2\% \\
		15 & 0.599994 & 1939426  & 669756   & 65.5\% \\
		\hline
		\multicolumn{2}{r}{Total:} & 4830908 & 1771223 & 63.3\% \\
		\hline
	\end{tabular}
	
	\vspace{6pt}
	\caption{Reduction of simplices by core computation for Vietoris-Rips complexes on a sample of 500 points on $S^2 \subset \mathbb{R}^3$ at various filtration values. The third column shows the number of simplices in the $3$-skeleton of $\mathrm{VR}(X,\varepsilon)$, the fourth column shows the number of simplices in the core, and the fifth column shows the percentage reduction.}
	\label{tab:reduction}
\end{table}

A problem that arises when attempting to replace $\mathrm{VR}(X,\varepsilon)$ with its core is that an inclusion of subcomplexes does not induce an inclusion of their cores. As a simple example, consider the complex $K$ consisting of a filled rectangle obtained by gluing two triangles along one of their edges, and let $L$ be the boundary of the rectangle. Clearly, the core of $L$ is $L$ itself, whereas the core of $K$ is a point. Therefore, if we apply the core construction to every complex in a filtration, we do not obtain a filtration. To circumvent this problem,  Boissonnat and Pritam \cite{BP20} defined a construction that turns the tower of cores into a filtration. But this also has a computational cost, so they approximate the PH by taking cores at some scale values (see \cite{BP20}). Other applications in data analysis based on strong collapses  can be found in \cite{WM13, WC13, WCK14}. 

Here we propose a different approach: we use the notion of strong collapse to define a subsampling technique. The advantage of this subsampling method over other alternatives is that it preserves both the global topology of the sample and its local structure: the underlying idea is to remove redundant points that are dominated by other points. This technique adapts equally well to samples with high concentration as to more sparse or heterogeneous samples (regions with varying density). The method has the following key properties: if a point $x$ is $\delta$-dominated by a point $y$, then every $\delta$-neighbor of $x$ is also a neighbor of $y$, ensuring preservation of local information. The reduction is adaptive to density since, by construction, it removes many points in dense regions while preserving structure in sparse regions.

 We will see that, given $\delta>0$, the $\delta$-core is unique up to equivalence (see Section \ref{sectiondeltacore} below). In particular, any two $\delta$-cores (for the same parameter $\delta$) have the same number of points and share similar structures. The complexity of the construction algorithm is $O(n\log n + n\bar{k})$ for computing neighborhoods using spatial data structures in low-dimensional Euclidean spaces, plus $O(nk_{\max}^2)$ per iteration for removing dominated vertices, where $n$ is the number of points, $\bar{k}$ is the average $\delta$-degree and $k_{\max}$ is the maximum $\delta$-degree. In practice, the empirical running time is close to $O(n\log n)$ for geometric point clouds in low dimensions. This preprocessing step substantially reduces the number of simplices in the Vietoris-Rips complex—often by millions when computing homology in degree greater than 1 for datasets with more than 1000--2000 points—thereby enabling significantly more efficient persistent homology computation. The construction algorithm and its pseudocode can be found in Section~\ref{sectiondeltacore}. In Section~\ref{sectionexperiments}, we present experimental results on both synthetic and real datasets, comparing the $\delta$-core method with the original sample and with other established subsampling techniques.

The scripts used to implement the $\delta$-core algorithm, generate the experimental results, and produce the figures and tables in this paper are publicly available at
\url{https://github.com/gabrielminian-dm/Delta-core-subsampling-strong-collapses-and-TDA}.

\section{Strong collapses and cores of simplicial complexes}\label{sectionstrong}

We recall first the basic constructions and results on strong collapses of simplicial complexes. For more details on strong homotopy types, including the relationship between cores and nerves of complexes, and applications to topology, combinatorics and the Evasiveness conjecture, the reader may consult \cite{BM11}. All the simplicial complexes that we deal with are assumed to be finite.

Let $K$ be a simplicial complex and let $\sigma$ be a simplex in $K$. Recall that the closed star of $\sigma$, denoted by $st(\sigma)$, is the subcomplex of simplices $\tau$ such that $\tau \cup \sigma$ is a simplex of $K$. The \textit{link} $lk(\sigma)$ is the subcomplex of $st(\sigma)$ of simplices disjoint from $\sigma$.  We denote by $K\smallsetminus v$ the \textit{deletion} of the vertex $v$, which is the full subcomplex of $K$ spanned by the vertices different from $v$. Note that  $K\smallsetminus v$ is obtained from $K$ by removing all the simplices containing $v$.

\begin{definition}
	Let $K$ be a simplicial complex and let $v\in K$ be a vertex. We say that there is an \textit{elementary strong collapse} from $K$ to $K\smallsetminus v$ if $lk(v)$ is a simplicial cone $v'L$ with apex $v'$. In this case we say that $v$ is \textit{dominated} (by $v'$) and we denote $K \esc K\smallsetminus v$. 
\end{definition}

\begin{remark}
	Note that a vertex $v$ is dominated by a vertex $v'\neq v$ if and only if every maximal simplex that contains $v$ also contains $v'$.
\end{remark}
	
	Figure \ref{fig:sc} shows an elementary strong collapse. The vertex $v$ is dominated by $v'$.
	
	\begin{figure}[htbp]
		\centering
		\includegraphics[scale=0.6]{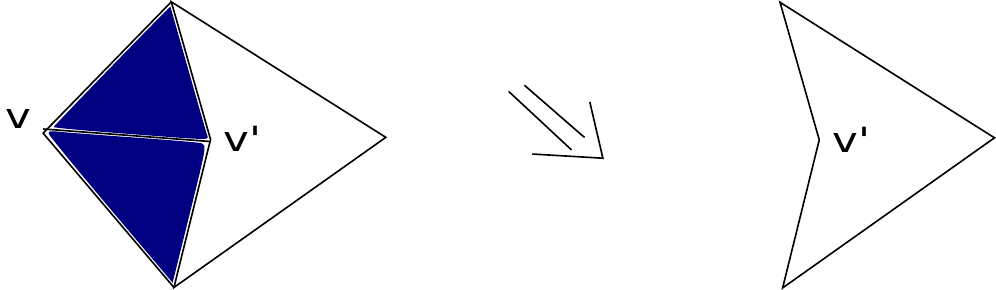}
		\caption{An elementary strong collapse.}
		\label{fig:sc}
	\end{figure}
	
	We say that there is a \textit{strong collapse} from a simplicial complex $K$ to a subcomplex $L$ if there exists a sequence of elementary strong collapses that starts in $K$ and ends in $L$. We write $K \scp L$.
	
\begin{remark}	
	If $v$ is dominated by $v'$, the map $r:K\to K\smallsetminus v$ which is the identity in $ K\smallsetminus v$ and such that $r(v)=v'$ is a simplicial map. Moreover, it induces a strong deformation retract. 
\end{remark}

\begin{proposition}
	If $K\scp L$ then $K$ simplicially collapses to $L$ and, in particular, the inclusion $i: L \hookrightarrow K$ is a strong deformation retract. A simplicial retraction $r:K\to L$ which is a homotopy inverse to $i$ is defined by iterating the retraction maps of the previous remark.
	\end{proposition}

We say that a simplicial complex $K$ is minimal if it has no dominated vertices. A \textit{core} of a simplicial complex $K$ is a minimal subcomplex $L\subseteq K$. Note that one can obtain a core of $K$ by removing dominated vertices one by one. The following result, proved in \cite{BM11}, asserts that the ordering in which we remove the dominated vertices is irrelevant since the core is unique up to isomorphism.

\begin{thm}\label{unique}
	Every simplicial complex has a core and it is unique up to isomorphism.
	\end{thm}
	
\subsection*{Cores of flag complexes}

The most commonly used simplicial complexes in persistent homology are Vietoris-Rips complexes. Given a point cloud (finite metric space) $(X,d)$ the Vietoris-Rips complex for a parameter $\varepsilon>0$ is the simplicial complex $\mathrm{VR}(X,\varepsilon)$ with vertex set $X$ and whose simplices are the finite sets $\sigma=\{x_0,\ldots,x_r\}$ such that $d(x_i,x_j)\leq \varepsilon$ for all $0\leq i,j\leq r$. By construction, Vietoris-Rips complexes are \textit{flag} complexes. Recall that a simplicial complex $K$ is flag if it satisfies the following condition: every finite subset of vertices $\sigma=\{v_o,\ldots,v_r\}$ is a simplex of $K$ if and only if every pair of vertices in $\sigma$ is a simplex (edge) of $K$. Equivalently, $K$ is flag if it is the clique complex of its $1$-skeleton (viewed as a graph).

\begin{remark}\label{domflag}
It is easy to see that, if $K$ is a flag complex, a vertex $v$ is dominated by a vertex $v'$ if and only if for every vertex $w$ such that $\{v,w\}\in K$, we have $\{v',w\}\in K$. Equivalently, $v$ is dominated by $v'$ if and only if the closed neighborhood $N[v]$ of $v$ in the $1$-skeleton of $K$ (viewed as a graph) is contained in the closed neighborhood $N[v']$ of $v'$. Note also that, in this case, $K\smallsetminus v$ is also flag, which implies that the cores of flag complexes are also flag.
\end{remark}

\section{Persistent homology}
We recall very briefly the basic notions on persitent homology and refer the reader to the excellent books \cite{EH08,G14,RB19,BCY18,Z09} for more details. We use simplicial homology over a field $\mathbb{F}$. A filtration is a nested sequence of simplicial complexes $K_0\subseteq K_1\subseteq K_2\subseteq\ldots\subseteq K_m$. For every degree $q\geq0$, if we apply the homology functor $H_q$, we obtain a sequence of $\mathbb{F}$-vector spaces 
$$H_q(K_0)\to H_q(K_1)\to H_q(K_2)\to\ldots\to H_q(K_m)$$
which is called a persistence module. For every $i<j$ there is a linear map $f_{ij}:H_q(K_i)\to H_q(K_j)$ induced by the inclusion $K_i\subseteq K_j$. Any (finite) persistence module can be decomposed as a multiset of intervals of the form $[i,j)$. The intervals can be viewed as points $(i,j) \in \R\times (\R\cup\{\infty\})$. This multiset is called the persistence diagram. Each interval (point) $[i,j)$ corresponds to a cycle which appears at $i$ (birth) and dissapears at $j$ (death). The points $(i,j)$ far from the diagonal are those that persist for a long time.

Given a point cloud $(X,d)$, one can take, for each $\varepsilon>0$, $K_{\varepsilon}=\mathrm{VR}(X,\varepsilon)$. Recall that $\mathrm{VR}(X,\varepsilon)$ is the Vietoris-Rips complex of $X$ with parameter $\varepsilon$, whose simplices are the finite subsets $\{x_0,\ldots,x_r\}\subseteq X$ with $d(x_k,x_l)\leq \varepsilon$. Note that, if $\varepsilon\leq \varepsilon'$ then $K_{\varepsilon}\subseteq K_{\varepsilon'}$, so the Vietoris-Rips complexes form a filtration, which is finite (since $X$ is finite). The persistent homology of the sample $X$ is the PH of its associated Vietoris-Rips filtration.

\subsection*{Distances between persistence diagrams}

To compare persistence diagrams, we will use two metrics, which are standard in topological data analysis. The \textit{bottleneck distance} between two persistence diagrams $D$ and $D'$ is defined as
\[
d_B(D, D') = \inf_{\gamma} \sup_{p \in D} \|p - \gamma(p)\|_\infty,
\]
where the infimum is taken over all matchings $\gamma: D \to D'$ (including the diagonal), and $\|\cdot\|_\infty$ denotes the supremum norm. The \textit{Wasserstein distance} of order $q$ is defined as
\[
W_q(D, D') = \left(\inf_{\gamma} \sum_{p \in D} \|p - \gamma(p)\|_\infty^q\right)^{1/q}
\]
In this article we will use $q=1$ (the $1$-Wasserstein distance). The bottleneck distance captures the maximum mismatch between paired features, making it particularly sensitive to the most significant differences. The Wasserstein distance accounts for all features cumulatively and is more sensitive to the overall distribution of the topological features.

\section{The $\delta$-core subsampling}\label{sectiondeltacore}

Given a point cloud $X$ endowed with a metric $d$, a parameter $\delta> 0$ and a point $x\in X$, we define the $\delta$-neighborhood of $x$ as the subset $\nd(x)=\{y\in X \ |\ d(x,y)\leq \delta\}$.

\begin{definition}
	Let $(X,d)$ be a finite metric space. Fix a parameter $\delta>0$. We say that a point $x\in X$ is dominated by a point $x'\neq x$ if $\nd(x)\subseteq \nd(x')$.
	\end{definition}
	
\begin{remark}
	Note that if $x$ is dominated by $x'$, then $d(x,x')\leq \delta$ and every point $y$ that is $\delta$-close to $x$ (i.e $d(y,x)\leq\delta$) is also $\delta$-close to $x'$.
\end{remark}

We define a $\delta$-core of $(X,d)$ similarly as we did for simplicial complexes.

\begin{definition}
	A $\delta$-core of a point cloud $(X,d)$ is a subsample with no dominated points.
\end{definition}

\subsection*{Construction of the $\delta$-core} Starting with a point cloud $X$ and a fixed parameter $\delta>0$, we look for dominated points. If $x_1\in X$ is a dominated point of $X$, we remove it from the sample and obtain a subsample $X-\{x_1\}$. We repeat this procedure until we reach a subsample without dominated points. In each step $i$ we remove a point $x_i$ that is dominated in $X-\{x_1,\ldots,x_{i-1}\}$ (with the metric $d$ inherited from $X$).

The $\delta$-core that we construct with this procedure depends, of course, on the points that we choose to remove. But, similarly as in the simplicial complex case, two different $\delta$-cores are equivalent in the following sense.

\begin{definition}
	Given a point cloud $(X,d)$ and a parameter $\delta$, we say that two subsamples $Y,Z\subseteq X$ are \textit{$\delta$-equivalent} if there is a biyection $\varphi:Y\to Z$ such that, for every $y,y'\in Y$, $d(y,y')\leq \delta$ if and only if $d(\varphi(y),\varphi(y')\leq\delta$.
 \end{definition}
 
\begin{thm}
Given $(X,d)$ and a parameter $\delta$, the $\delta$-core of $X$ is unique up to $\delta$-equivalence.
\end{thm}
		
\begin{proof}
	Consider the Vietoris-Rips complex $K=\mathrm{VR}(X,\delta)$. Recall that the vertex set of $K$ is $X$ and a subset $\sigma=\{x_0,\ldots,x_r\}$ is a simplex of $K$ if and only if $d(x_i,x_j)\leq \delta$ for every  $i,j$. By Remark \ref{domflag}, $x$ is a dominated point of $X$ if and only if it is a dominated vertex of the simplicial complex $K$. Note also that $K\smallsetminus x=\mathrm{VR}(X-\{x\},\delta)$. This implies that a $\delta$-core of the point cloud $X$ corresponds to the set of vertices of a core of $K$. Now the results follows from Theorem	\ref{unique}.
	\end{proof}
	
In particular all possible $\delta$-cores of $X$ have the same cardinality and similar shape. See [Algorithm 1] below for the pseudo-code of the construction.

The computational cost of the $\delta$-core algorithm has two components. 
First, computing all $\delta$-neighborhoods (lines 5--7 in [Algorithm 1]) requires determining, 
for each point, which other points lie within distance $\delta$. A naive implementation 
evaluates all $O(n^2)$ pairwise distances. However, for low-dimensional Euclidean data (dimension $d \lesssim 20$) using spatial data structures 
such as KD-trees  or ball trees, this step 
can be performed in $O(n \log n + n\bar{k})$ time, where $\bar{k}$ is the average 
$\delta$-degree.

In the second phase (lines 10--22), the algorithm repeatedly removes dominated vertices. 
Checking whether $N(x)\subseteq N(y)$ for a pair of adjacent vertices requires 
$O(k_{\max})$ time, and each vertex compares itself to at most $k_{\max}$ 
neighbors. Therefore one full sweep costs $O(n k_{\max}^2)$. 
In the worst case the algorithm may perform $O(n)$ sweeps, leading to a 
theoretical upper bound $O(n^2 k_{\max}^2)$. 
However, for geometric point clouds in low dimensions, $k_{\max}$ is typically 
small and many vertices are removed simultaneously, so the number of sweeps is 
small in practice. With efficient neighborhood computation, the empirical 
running time is close to $O(n \log n)$.

This preprocessing step can significantly reduce the size of the 
Vietoris-Rips complex. Persistent homology via matrix reduction has 
complexity $O(s^3)$, where $s$ is the number of simplices. 
For $\mathrm{VR}(X,\varepsilon)$ up to dimension $d$, the worst-case number of 
simplices is $s = O(n^{d+1})$. 
If the core preserves only a fraction $\alpha < 1$ of the vertices, the number 
of simplices scales as $\alpha^{d+1}$, and the complexity improves from
$O(s^3)$ to $O\big((\alpha^{d+1} s)^3\big).$

\begin{algorithm}[H]
	\caption{$\delta$-Core$(X)$}
	\begin{algorithmic}[1] 
		\State \textbf{Input:} Point cloud $X=\{x_1,\dots,x_n\}$, parameter $\delta>0$
		\State \textbf{Output:} The $\delta$-core
		\State
		\State // Precompute $\delta$-neighborhoods 
		\For{$i=1$ to $n$}
		\State $N_\delta(x_i) \gets \{\, j \mid \mathrm{dist}(x_i,x_j)\le \delta\,\}$
		\EndFor
		\State
		\State $\mathrm{Active} \gets \{1,\dots,n\}$
		\Repeat
		\State $\mathrm{ToRemove} \gets \varnothing$
		\ForAll{$i\in\mathrm{Active}$}
		\If{$i\in\mathrm{ToRemove}$} \State \textbf{continue} \EndIf
		\State $N_i \gets N_\delta(x_i)\cap\mathrm{Active}$
		\ForAll{$j\in N_i$}
		\If{$j = i$} \State \textbf{continue} \EndIf
		\If{$j\in\mathrm{ToRemove}$} \State \textbf{continue} \EndIf
		\State $N_j \gets N_\delta(x_j)\cap\mathrm{Active}$
		\If{$N_i \subseteq N_j$} \Comment{$x_i$ is dominated by $x_j$}
		\State $\mathrm{ToRemove} \gets \mathrm{ToRemove}\cup\{i\}$
		\State \textbf{break} \Comment{stop checking other $j$ for this $i$}
		\EndIf
		\EndFor
		\EndFor
		\State $\mathrm{Active}\gets\mathrm{Active}\setminus\mathrm{ToRemove}$
		\Until{$\mathrm{ToRemove}=\varnothing$}
		\State \Return $\mathrm{Core}\gets\text{sorted}(\mathrm{Active})$
	\end{algorithmic}
\end{algorithm}

\subsection*{The parameter $\delta$} 

The choice of the parameter $\delta$ depends on the characteristics of the data. If $\delta$ is too large, the subsample size will be small and topological information may be lost. For instance, when $\delta$ approaches the diameter of the dataset, the core tends to a single point (the core of a simplex is a vertex). Conversely, if $\delta$ is too small, the reduction will be minimal and the computational gain negligible. 

For relatively uniform samples, a value around the 15th to 20th percentile of pairwise distances is typically appropriate, yielding substantial computational improvement, especially when the sample contains more than 1000 points and persistent homology is computed up to degree 2 or higher. For sparser or more heterogeneous samples, a lower percentile such as the 7th to 10th may be preferable to preserve topological features in regions of varying density. In practice, the optimal choice of $\delta$ can be determined by examining the reduction rate and comparing persistence diagrams at different percentile values.

Figure~\ref{fig:delta_cores} displays $\delta$-core subsamples for different values of $\delta$ on a heterogeneous sample of 1500 points in the unit cube $[0,1]^3 \subset \mathbb{R}^3$. In this particular case, as $\delta$ increases, the size of the $\delta$-core decreases. This is, in general, the typical behavior; however, depending on the structure of the sample, increasing $\delta$ may also lead to an increase in the subsample size. Of course, when $\delta$ is sufficiently large, the $\delta$-core consists of a single point, and when $\delta$ is very small, the $\delta$-core coincides with the entire sample.

\begin{figure}[htbp]
	\centering
	\includegraphics[width=\textwidth]{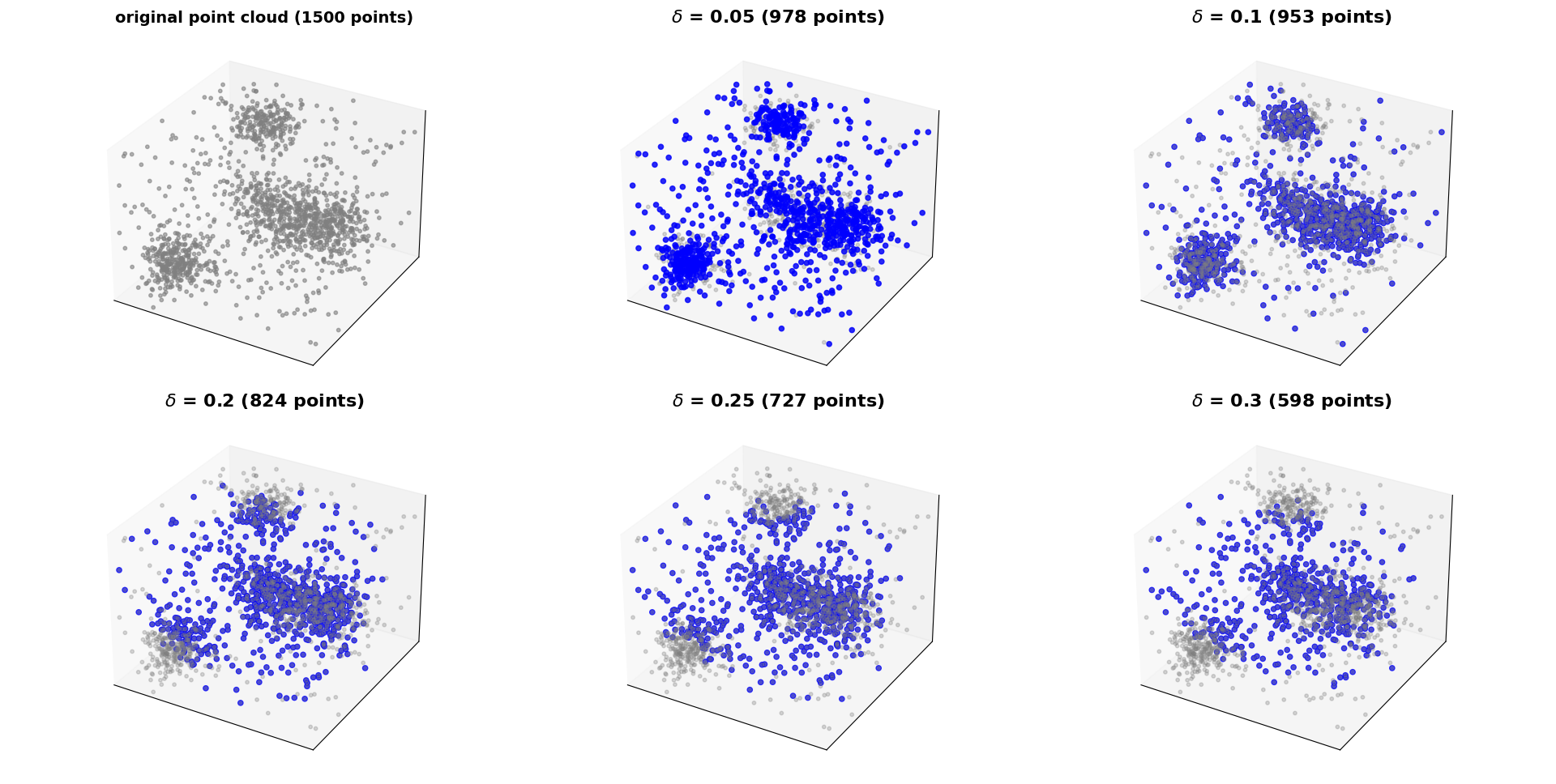}
	\caption{$\delta$-core subsamples for different values of $\delta$.}
	\label{fig:delta_cores}
\end{figure}

\section{Experimental results}\label{sectionexperiments}

We present experimental results on both synthetic and real datasets. First, we show two examples with synthetic data: the heterogeneous point cloud in the unit cube from the previous section, and a more homogeneous sample of 2000 points on a 2-dimensional torus in $\mathbb{R}^3$. We evaluate computation times and distances between persistence diagrams. We also contrast the $\delta$-core subsampling with FPS (maxmin) subsampling, using the same number of points, and benchmark our construction against the witness complex (with randomly chosen landmarks) Finally, we present two additional examples with real heterogeneous datasets and compare our method with FPS subsampling and a recently introduced outlier-robust subsampling method \cite{S23}. In all cases, our method achieves better approximations (both in bottleneck and Wasserstein distances).

We implement our algorithm in Python and use Ripser.py \cite{TSB18} for persistent homology and GUDHI \cite{GUDHI15} for witness complexes. All experiments are performed on a standard desktop machine equipped with an 11th Gen Intel Core i7 (2.50 GHz) CPU and 16 GB of RAM.

\subsection*{Points in $3$-dimensional cube} 

We first compare the heterogeneous sample of 1500 points shown in Figure~\ref{fig:delta_cores} using $\delta=0.05$ and $\delta=0.3$. Figure~\ref{fig:cube} shows the persistence diagrams computed with a threshold of 0.8. For $\delta=0.05$, the $\delta$-core contains 978 points and requires $0.36s$  for construction and $42.78s$  for persistence computation, achieving a $5.04\times$ speedup compared to the $217.55s$  required for the original sample. For $\delta=0.3$, the $\delta$-core contains 598 points and requires $10.09s$  for construction and $16.07s$  for persistence computation, achieving an $8.32\times$ speedup.

\begin{figure}[htbp]
	\centering
	\includegraphics[width=\textwidth]{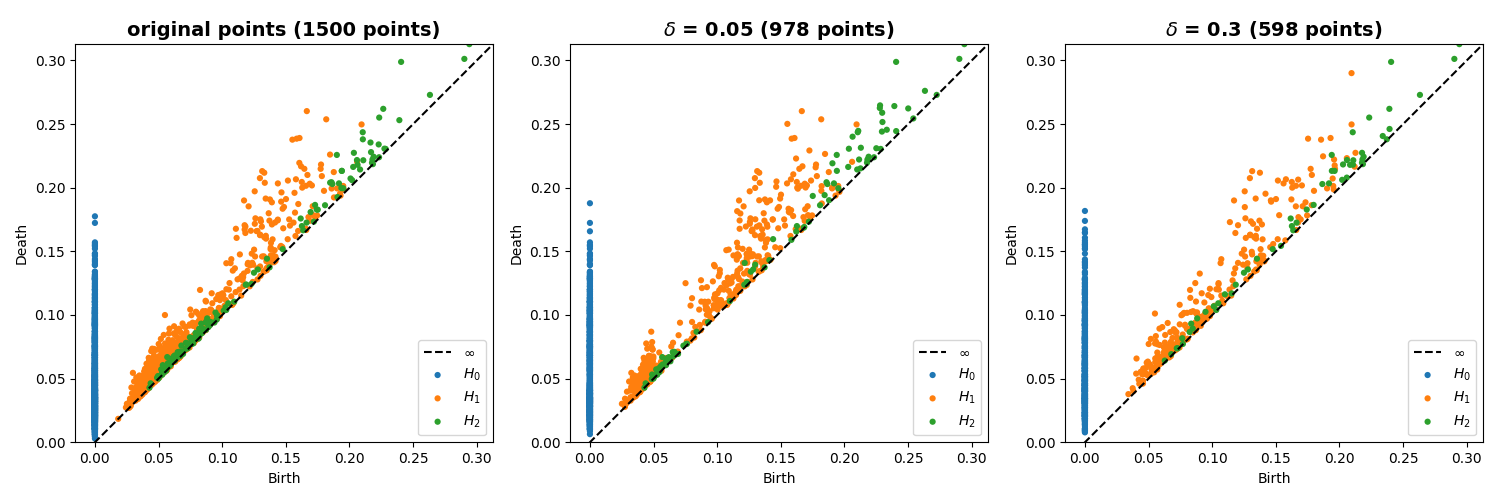}
	\caption{Points in the $3$-dimensional cube: Comparison of persistence diagrams. Original sample and $\delta$-cores with $\delta=0.05$ and $0.3$.}
	\label{fig:cube}
\end{figure}

We now compare, for the same point cloud, the $\delta$-core subsampling with $\delta=0.25$ against FPS (maxmin) subsampling using the same number of points. Figure~\ref{fig:cube-comparison} shows the persistence diagrams. Although the bottleneck distances between the original sample and the FPS subsample are similar to those between the original sample and the $\delta$-core subsample, the Wasserstein $L^1$ distances between the original sample and the $\delta$-core subsampling are smaller in all three dimensions. Table~\ref{tab:cube_wasserstein} shows the Wasserstein $L^1$ distances for $H_1$ and $H_2$ between the original sample and both subsampling methods.

\begin{figure}[htbp]
	\centering
	\includegraphics[width=\textwidth]{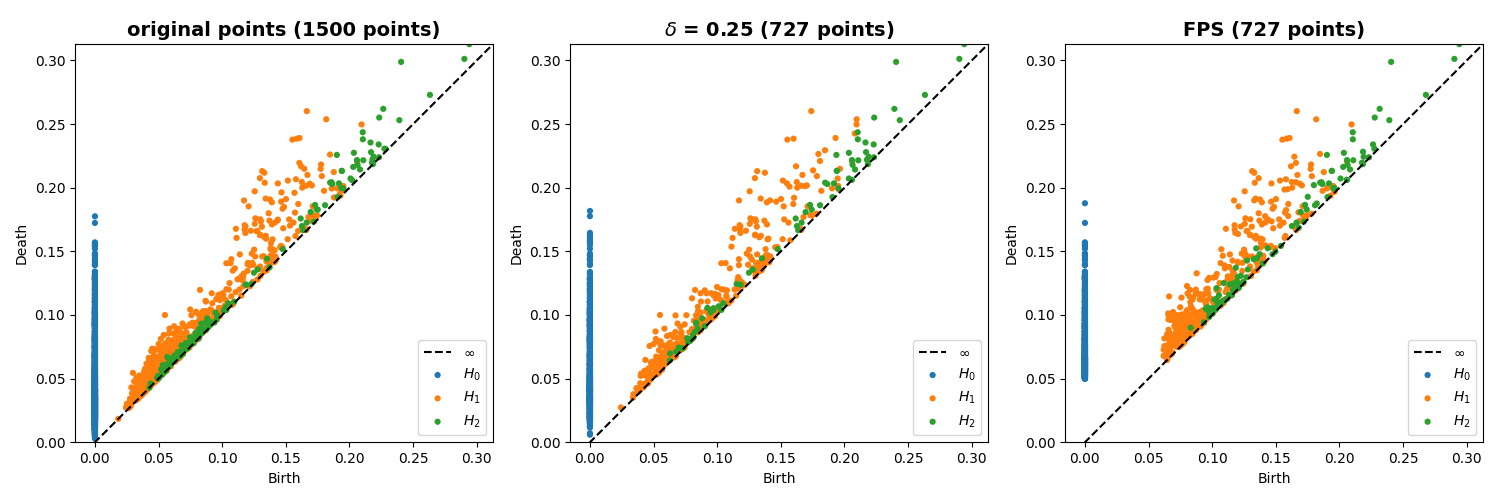}
	\caption{Points in the cube: Comparison of persistence diagrams. Original sample, $\delta$-core with $\delta=0.25$, and FPS.}
	\label{fig:cube-comparison}
\end{figure}

\begin{table}[htbp]
	\centering
	\begin{tabular}{lcc}
		\hline
		
		Method & $H_1$ & $H_2$ \\
		\hline
		$\delta$-core (727 points) & 2.569577 & 0.228622 \\
		FPS (727 points) & 3.436561 & 0.432619 \\
		\hline
	\end{tabular}
	\vspace{6pt}
	\caption{Points in the cube: Wasserstein $L^1$ distances between original sample and subsamples.}
	\label{tab:cube_wasserstein}
\end{table}

\subsection*{Torus in $\mathbb{R}^3$} 

We now consider a more homogeneous sample of 2000 points on a 2-dimensional torus in $\mathbb{R}^3$. The points are generated by sampling the angular parameters uniformly at random with additional Gaussian noise. We use $\delta=1.4$ and compute persistent homology up to degree 2 with Ripser using a threshold of 2. Figure~\ref{fig:torus_points} shows the original sample and the $\delta$-core subsampled set (in blue), which contains 1527 points. Figure~\ref{fig:torus_diagrams} shows the persistence diagrams. The bottleneck distances between the original sample and the subsample are 0.050543 in degree 1 and 0.035024 in degree 2.

\begin{figure}[htbp]
	\centering
	\includegraphics[width=\textwidth]{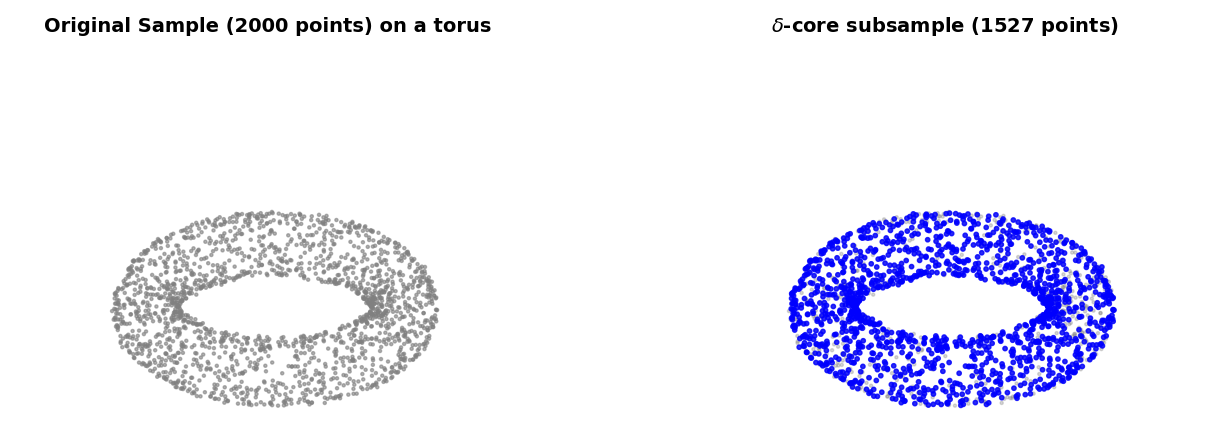}
	\caption{2000 points sampled on the 2-torus in $\mathbb{R}^3$ in grey. The $\delta$-core for $\delta=1.4$ in blue (1527 points).}
	\label{fig:torus_points}
\end{figure}

\begin{figure}[htbp]
	\centering
	\includegraphics[width=\textwidth]{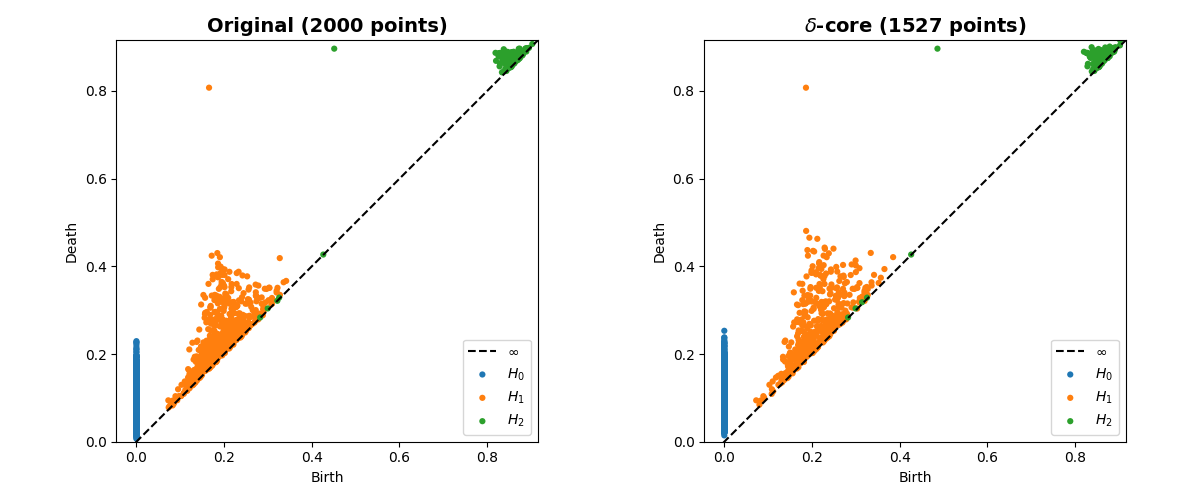}
	\caption{Torus: Comparison of persistence diagrams. Original sample and $\delta$-core with $\delta=1.4$.}
	\label{fig:torus_diagrams}
\end{figure}

We also try with a sample of 700 points on the same torus, now with $\delta=0.4$, obtaining a $\delta$-core with 436 points, and compare with the PH of the witness complex with a number of landmarks equal to the size of the $\delta$-core (randomly chosen from the sample). The witness complex is computed in Python with \texttt{GUDHI}, using the function \texttt{EuclideanWitnessComplex.create\_simplex\_tree} and the same filtration range as in the Vietoris-Rips computations. For both methods we compute persistence up to dimension 2 and compare their diagrams with those of the full dataset using bottleneck distances.

Table~\ref{tab:compare-dc-witness} shows the results. In this example the $\delta$-core gives slightly smaller bottleneck distances in both $H_1$ and $H_2$. The main contrast is in running time: the whole $\delta$-core pipeline (construction plus persistence) took about $0.25s$ , while the witness complex took over $1000s$  on the same machine. 

\begin{table}[h]
	\centering
	
	\label{tab:compare-dc-witness}
	\begin{tabular}{lcc}
		\hline
		Method & $H_1$ & $H_2$ \\
		\hline
		$\delta$-core (436 points) & $0.1814$ & $0.0889$ \\
		Witness (436 random landmarks) & $0.2605$ & $0.0965$ \\
		\hline
	\end{tabular}
	\vspace{6pt}
	\caption{Torus: bottleneck distances to the original persistence diagram.}
\end{table}

\subsection*{Breast Cancer Wisconsin Dataset}

We now analyze the Breast Cancer Wisconsin (Diagnostic) dataset from the UCI Machine Learning Repository, available through scikit-learn. This dataset contains 569 samples in $\mathbb{R}^{30}$, where each point corresponds to measurements computed from digitized fine-needle aspirate (FNA) images of breast masses collected at the University of Wisconsin Hospitals. The features describe various morphological characteristics of the cell nuclei present in the images. Before analysis, we standardize the data to have zero mean and unit variance.

We compare the $\delta$-core subsampling against FPS subsampling using $\delta$ set to the 10th percentile of pairwise distances. Figure~\ref{fig:breast_cancer} shows the persistence diagrams for the original sample and both subsampling methods, each using 429 points. Table~\ref{tab:breast_cancer_bottleneck} shows the bottleneck distances for $H_1$ and $H_2$ between the original sample and both subsamples. The $\delta$-core subsample achieves a better approximation of the original topological features, with bottleneck distances roughly half those of FPS in both dimensions.

\begin{figure}[htbp]
	\centering
	\includegraphics[width=\textwidth]{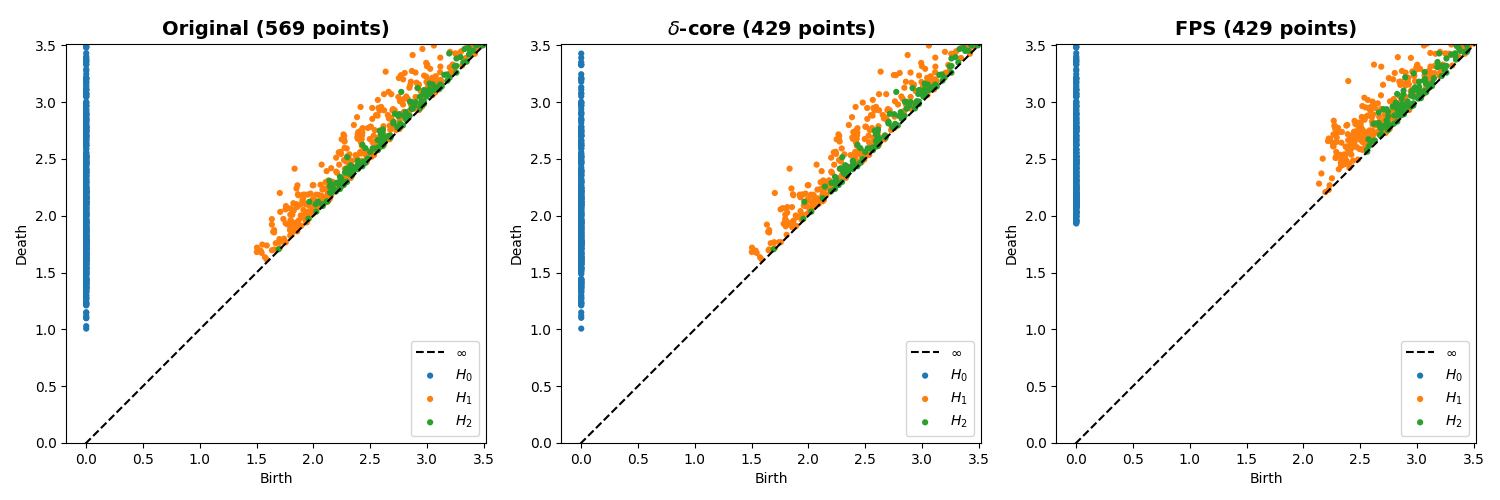}
	\caption{Breast cancer dataset: comparison of persistence diagrams. Original sample, $\delta$-core, and FPS.}
	\label{fig:breast_cancer}
\end{figure}

\begin{table}[htbp]
	\centering
	\begin{tabular}{lcc}
		\hline
		Method & $H_1$ & $H_2$ \\
		\hline
		$\delta$-core (429 points) & 0.106293 & 0.063485 \\
		FPS (429 points) & 0.290054 & 0.108560 \\
		\hline
	\end{tabular}
	\vspace{6pt}
	\caption{Breast cancer dataset: bottleneck distances between original sample and subsamples.}
	\label{tab:breast_cancer_bottleneck}
\end{table}
\subsection*{Head CT Medical Images Dataset}

We analyze head CT images from the HeadCT subset of the MedMNIST dataset \cite{medmnistv2}. MedMNIST provides standardized $64\times 64$ grayscale images derived from public medical imaging datasets, including CQ500 for the HeadCT class. We randomly select 1000 Head CT images from this collection. Each image, originally in $\mathbb{R}^{4096}$ (corresponding to $64\times 64$ pixels), is first standardized to have zero mean and unit variance, and then reduced to 40 dimensions using PCA, which simplifies the analysis while keeping the essential structure of the images.

We compare the $\delta$-core subsampling (using $\delta$ equal to the 10th percentile of pairwise distances) against FPS and the outlier-robust subsampling method of Stolz~\cite{S23}. The algorithm for Stolz's method was obtained from the author's GitHub repository (see \cite{S23}). Figure~\ref{fig:headct} shows the persistence diagrams for the original sample and the three subsampling methods, each using 712 points. Table~\ref{tab:headct_distances} shows both bottleneck and Wasserstein distances for $H_1$ and $H_2$ between the original sample and each subsample. As can be observed in the table, the $\delta$-core method achieves better approximations than both FPS and Stolz's method in all metrics. Regarding computational cost, we measure the construction time for each subsampling method plus the subsequent persistent homology computation on our machine. FPS is the fastest method ($0.19s$ total), followed by $\delta$-core ($8.12s$ total), while the outlier-robust method requires more time ($109.70s$ total). The $\delta$-core construction takes $7.45s$, while the persistent homology computation on the resulting subsample takes $0.67s$.

\bigskip

\begin{figure}[htbp]
	\centering
	\includegraphics[width=\textwidth]{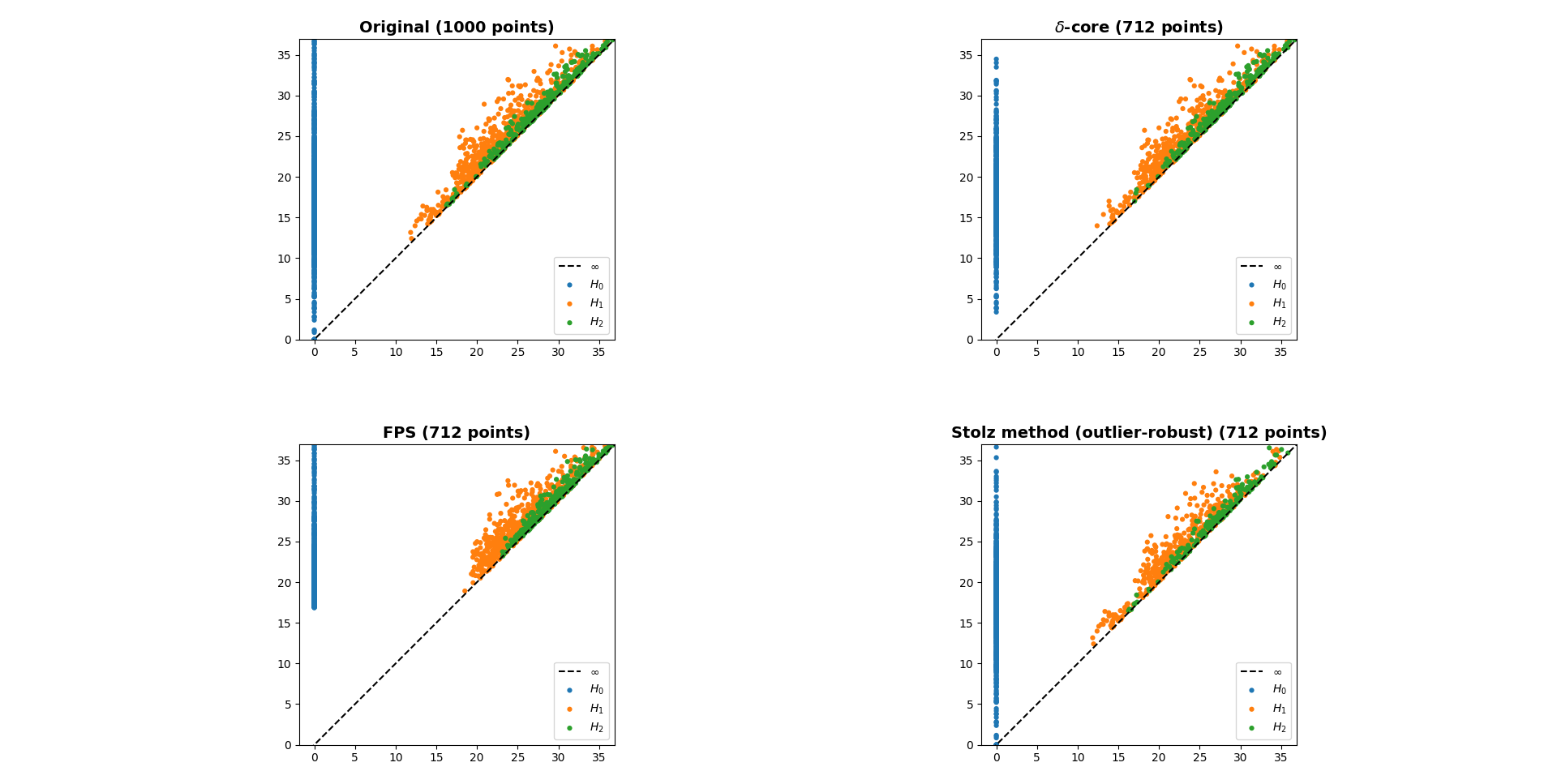}
	\caption{Head CT dataset: comparison of persistence diagrams.}
	\label{fig:headct}
\end{figure}

\begin{table}[htbp]
	\centering
	\begin{tabular}{lcccc}
		\hline
		& \multicolumn{2}{c}{Bottleneck} & \multicolumn{2}{c}{Wasserstein} \\
		Method & $H_1$ & $H_2$ & $H_1$ & $H_2$ \\
		\hline
		$\delta$-core (712 points) & 1.594166 & 0.655020 & 164.267939 & 18.350527 \\
		Outlier-robust (712 points) & 2.650137 & 1.086998 & 242.289767 & 66.919088 \\
		FPS (712 points) & 1.647661 & 1.038147 & 370.236394 & 81.843747 \\
		\hline
	\end{tabular}
	\vspace{6pt}
	\caption{Head CT dataset: bottleneck and Wasserstein distances between original sample and subsamples.}
	\label{tab:headct_distances}
\end{table}

\section{Conclusions}

We presented an effective subsampling method based on the topological reduction technique of strong collapse. We contrasted our construction with other known methods and observed that, in general, it produces more accurate subsamples and better approximates persistent homology, while maintaining a relatively low computational cost compared with other alternative methods.

\end{document}